\pdfoutput=1
\documentclass[journal]{IEEEtran}
\usepackage{cite}

\ifCLASSINFOpdf
\else
\fi

\usepackage{stfloats}
\usepackage{graphicx}
\usepackage{dcolumn}
\usepackage{bm}
\usepackage[mathlines]{lineno}

\usepackage{amsmath, amsfonts, amsbsy, amssymb, amsthm}
\usepackage[utf8]{inputenc}
\usepackage[T1]{fontenc}
\usepackage{mathptmx}

\usepackage{array,url,kantlipsum}
\usepackage{multicol,lipsum,xparse}

\usepackage{url}
\usepackage{setspace}
\usepackage[table]{xcolor}
\usepackage{verbatim}
\usepackage[export]{adjustbox}
\usepackage{subcaption}

\usepackage{color, colortbl}

\definecolor{c1}{RGB}{92,89,191}
\definecolor{c2}{RGB}{191,92,98}
\definecolor{c3}{RGB}{98,191,92}

\newtheorem{lemma}{Lemma}
\newtheorem{remark}{Remark}
\newtheorem{theorem}{Theorem}
\newtheorem{definition}{Definition}
\newtheorem{assumption}{Assumption}

\usepackage[utf8]{inputenc}
\usepackage{multirow}

\author{Kshitij Bhatta \IEEEmembership{Student Member, IEEE}, Majeed Hayat \IEEEmembership{Fellow, IEEE} and Francesco Sorrentino \IEEEmembership{Senior Member, IEEE} %
    \thanks{This work was supported by the National Science Foundation through Grant No. 1727948 and Grant CRISP-1541148.}
    \thanks{K. Bhatta, is with the Department of Mechanical Engineering, Univeristy of New Mexico, Albuquerque, NM 87131 USA}
    \thanks{M. Hayat, is with the Department of Electrical and Computer Engineering, Marquette University, Milwakee, WI 53233 USA}
    \thanks{F. Sorrentino is with the Department of Mechanical Engineering, University of New Mexico, Albuquerquem NM 87131 USA (e-mail: fsorrent@unm.edu).}}

\title{Modal Decomposition of the Linear Swing Equation \\ in Networks with Symmetries}

\date{August 2020}
\begin{document}
\maketitle

\begin{abstract}
 Symmetries are widespread in physical, technological, biological, and social systems and networks, including power grids. The swing equation is a classic model for the dynamics of powergrid networks. The main goal of this paper is to explain how network symmetries affect the swing equation transient and steady state dynamics. We introduce a modal decomposition that allows us to study transient effects, such as the presence of overshoots in the system response. This modal decomposition provides insight into the peak flows within the network lines and allows a rigorous characterization of the effects of symmetries in the network topology on the dynamics. Our work applies to both cases of homogeneous and heterogeneous parameters. Further, the model is used to show how small perturbations propagate in networks with symmetries.  Finally, we present an application of our approach to a large power grid network that displays symmetries. 
\end{abstract}

\section{Introduction}

Many papers have investigated network models which describe the dynamics of power grids  \cite{motter2002cascade,crucitti2004model,crucitti2004topological,van2007voltage,rohden2016cascading,rosas2007topological,albert2004structural,holmgren2006using,sole2008robustness,chassin2005evaluating,dwivedi2010analyzing,dwivedi2009identifying,guohua2008vulnerability,pahwa2010topological,pahwa2014abruptness,Syncagrearan}. 
Simplified models for the propagation of cascading failures on networks have been proposed in \cite{motter2002cascade,crucitti2004model,crucitti2004topological}. More realistic models for the propagation of cascading failures are based on the swing equation \cite{rohden2016cascading} or DC power flows \cite{carreras2002critical,pahwa2014abruptness}. Symmetries play a significant role in the study of networked systems. References \cite{okuda1991mutual, golubitsky1999symmetry,NC,blaha2019cluster,della2020symmetries,ACS,nicosia2013remote, whalen2015observability,cho2017stable,barrett2017equitable,sorrentino2019symmetries,S_nchez_Garc_a_2020}   have proposed tools based on graph theory and group theory to analyze the dynamics of complex networks with symmetries. The presence of symmetries in power grid networks has been documented in \cite{S_nchez_Garc_a_2020,della2020symmetries}. {Reference \cite{AranyaGTA} has analyzed how
network symmetries may affect synchronization modes of power
grids and suggested that symmetries may enhance the emergence of complete
synchronization.}


\color{black}

{Despite this previous work, it appears that the effects of the network symmetries on the dynamics of the swing equation have not been fully elucidated. In what follows we first provide a definition for network symmetries in the context of the swing equation, then we show how a reduced representation of the dynamics based on the so-called `quotient network' can be achieved both in the cases of homogeneous and heterogeneous parameters. Only a subset of the modes of the original network are \emph{inherited} by the quotient network. Neglecting the remaining modes leads to an approximation, which nonetheless can be quantified by considering an appropriately defined error dynamics and applying a modal analysis to the error dynamics.}

An important application of the study of network dynamics is network design. Knowledge of how the dynamics changes in response to structural and dynamical perturbations is key for designing resilient complex systems. A complex network approach to study vulnerabilities of power grids has been presented in \cite{arianos2009power}. { Further, several articles have studied the existence of vulnerabilities inherent to the power network structure \cite{Structuralvunera,ChenProbab,VulnerabilityBald,BistableZheng,BernsteinSensitivity,LargenetworkCascadesNishikawa,EmergentfailuresNesti}} and the importance of the transient dynamics in the propagation of failures has been emphasized in \cite{DIcascade}.

Here we propose a mathematical analysis of the classic swing equation based on a simplified description in terms of a network of coupled forced second order systems. Our analysis provides immediate understanding of the swing equation transient dynamics via a modal decomposition.  Further, it allows us to rigorously address the presence of network symmetries and their effects on the dynamics.

{The rest of this paper is organized as follows. Section II introduces the swing equation for a generic network and presents its modal decomposition. The effects of network symmetries are discussed in Sec.\ III. In the presence of symmetries, a lower dimensional dynamical representation based on the so-called quotient network is possible, either in the cases that the powers at the network nodes \textit{respect} the symmetries or not. In the latter case, the quotient description only provides an approximation of the full transient response. The deviation between the full network response and the quotient network response can be characterized in terms of a properly defined error dynamics. 
Effects of heterogeneity in the damping terms are presented in Sec.\ IV. Finally, the conclusions are given in Sec.\ V.}

\section{The Swing Equation and its Modal Decomposition}
The swing equation is a classic model for the dynamics of power grid networks. A power grid can be represented using an undirected graph such that each node can be either a generator (generating power) or a load (consuming power). Each node is regarded as a rotating machine (oscillator) and the presence of an edge between two nodes corresponds to the presence of a transmission line connecting them. The network connectivity is given by the symmetric matrix $\tilde{A}=\{ \tilde{A}_{ij} \}$, $\tilde{A}_{ij}=\tilde{A}_{ji} > 0$ if nodes $i$ and $j$ are connected, $\tilde{A}_{ij}=\tilde{A}_{ji}= 0$ otherwise.

      \begin{definition}
      {An undirected network} is defined by the set of nodes $\mathcal V=(1,2,...,n)$, $|\mathcal{V} |=n$ and the set of edges or lines $\mathcal{E}=\mathcal{V}\times \mathcal{V}, $ such that $(i,j) \in \mathcal{E}$,  $i \in \mathcal{V}$, $j \in \mathcal{V}$, if $\tilde{A}_{ij}=\tilde{A}_{ji}>0$.
      \end{definition}
      The state of each  oscillator $i=1,...,n$ is characterized by its node displacement $\theta_i$ and the nodal velocity $\omega_i = \dot{\theta_i}$ relative to a reference frequency $\omega_R$. $H_i$ is the inertia constant of oscillator $i$, $D_i$ is its damping constant and $b_i$ is the power generated (consumed) by the node. If we assume that $\omega_R \approx \omega$ and lossless transmission lines, we can write the swing equation \cite{powersync_latora, huang2019small},
      
      \begin{equation} \label{Aij}
   \frac{2H_i}{\omega_R}\ddot{\vartheta_i}+\frac{D_i}{\omega_R}\dot{\vartheta_i}= {b_i}-\sum_{{j=1,j\neq i}}^n {\tilde{A}_{ij}} \sin(\vartheta_i - \vartheta_j), \hspace{3mm} i=1,2,...n, 
 \end{equation}
$b_i>0$ for generators and $b_i<0$ for loads. Dividing both sides of Eq. \ref{Aij} by $\frac{2H_i}{\omega_R}$, we get
    \begin{equation}\label{interm}
      \ddot{\vartheta_i}+\frac{D_i}{2H_i}\dot{\vartheta_i}= \frac{\omega_R}{2H_i}{b_i}-\frac{\omega_R}{2H_i}\sum_{{j=1,j\neq i}}^n {\tilde{A}_{ij}} \sin(\vartheta_i - \vartheta_j), 
    \end{equation}
    {$i=1,2,...n.$} 
 
If we define, $\frac{D_i}{2H_i}=\gamma_i$, $J_i=\frac{\omega_R}{2H_i}$, ${p}_i=J_i b_i$, ${A}_{ij}=J_i \tilde{A}_{ij}$, Eq.\ \eqref{interm} becomes\\
 \begin{equation} \label{two}
       \ddot{\vartheta_i}=-\gamma_i\dot{\vartheta_i}+ p_i+  \sum_{{j=1,j\neq i}}^n {A}_{ij} \sin(\vartheta_j - \vartheta_i),  
    \end{equation}
$i=1,2,...n$. The flow in an edge/line $(i, j)$, with coupling ${A}_{ij}$ at time $t$ is given by \cite{DIcascade}:
\begin{equation}
    F_{ij}(t)={A}_{ij} \sin(\vartheta_j(t)-\vartheta_i(t)). 
    \label{flow}
\end{equation}

 \begin{definition}
 We say that the network is balanced if  $\sum_{i=1}^n b_i=0$. 
 \end{definition}
 Unless differently noted, we proceed under the assumption that the network is balanced.

In any large network,  including power-grid systems, there are redundancies in the form of symmetries \cite{okuda1991mutual, golubitsky1999symmetry,NC,della2020symmetries,ACS,nicosia2013remote, whalen2015observability,cho2017stable,barrett2017equitable}. 
\begin{definition}

A symmetry for the set of Eqs.\ \eqref{two} is a permutation matrix $P$ such that $P \pmb{\gamma}=\pmb{\gamma}$ 
and $P {A}={A} P$. 
 The automorphism group $\mathcal{G}$ is the set of all symmetries with the operation composition. The set of all symmetries in the group will only permute certain subsets of nodes (the \emph{orbits} or \emph{clusters}) among each other. The set of nodes $\mathcal{V}$ is partitioned
into $q$ disjoint subsets of nodes $\lbrace S_1, S_2,... S_q\rbrace$,  $\cup_{i=1}^q S_i=\mathcal{V}$, $S_i \cap S_j =\emptyset $ for $i \neq j$, with $n= n_i \sum_{i=1}^q n_i$ where $n_i=|S_i|$. 

The nodes in each subset are mapped into each other by application of one or more symmetries in $\mathcal{G}$; however, there is no symmetry in $\mathcal{G}$ that will map into each
other nodes in different subsets. We refer to
such subsets of nodes as `clusters’ or `orbits' of the automorphism group.
For a review of graph automorphisms, see \cite{S_nchez_Garc_a_2020}.

\begin{remark}
Consider  a permutation $P$ of the network nodes that satisfies $AP=PA$. Say $v,w \in \mathcal{V}$ two network nodes, call $v'$ $(w') \in \mathcal{V}$ the network node $v$ $(w)$ gets mapped to by application of the permutation $P$. It follows that $A_{vw}=A_{v'w'}$ and $A_{wv}=A_{w'v'}$.
\end{remark}
\end{definition}

\begin{lemma}
 A flow-invariant `synchronous' solution $\vartheta_i^*(t)=\vartheta^k(t)$ for all $i \in S_k$, $k=1,...,q$, is induced by the automorphism group $\mathcal{G}$. 
\end{lemma}
\begin{proof}Assume $\vartheta_i(0)=\vartheta^k(0)$ and $\dot{\vartheta}_i(0)=\dot{\vartheta}^k(0)$ for all $i \in S_k$, $k=1,...,q$. It follows that $\ddot{\vartheta}_i(0)=\ddot{\vartheta}^k(0)$ for all $i \in S_k$, $k=1,...,q$, from which the assertion follows.
\end{proof}

{Stability of the nonlinear swing equation (3) has been studied in \cite{vu2015lyapunov} by using a Lyapunov function approach. The conditions of \cite{vu2015lyapunov} can be directly applied to ensure convergence of the flow invariant solution on a {stable} fixed point $\vartheta_i^*=\vartheta^k$, for all $i \in S_k$, $k=1,...,q$.}

{
The linearized  swing equation, which models the propagation of small disturbances  (e.g., affecting the initial condition or affecting the power supplied/demanded at different nodes) \cite{tyloo2020primary}, is obtained by linearizing Eqs.\ (3) about the {stable} fixed point $\vartheta_i^*$, $i=1,...,n$,}
\begin{equation} \label{lse}
       \ddot{\theta_i}=-\gamma_i\dot{\theta_i}+  {p}_i+  \sum_{{j=1}}^n {L_{ij}} \theta_j,
    \end{equation}
$i=1,2,...n$, where the Laplacian matrix $L=\{L_{ij}\}$ has entries $L_{ij}=  {A}_{ij} \cos(\vartheta_j^*-\vartheta_i^*) -\delta_{ij}  \sum_j {A}_{ij} \cos(\vartheta_j^*-\vartheta_i^*)$,  and $\delta_{ij}$ is the Kronecker delta.{Each term $p_i$ on the right hand side of Eq.\ (5) effectively represents a small power deviation.  In the rest of this paper we will often approximate the nonlinear swing equation (3) with the linearized swing equation (5),  under the assumption of small power deviations.}  

\begin{definition} \label{defsym}
 A symmetry for the set of Eqs.\ \eqref{lse} is a permutation matrix $P$ such that $P \pmb{\gamma}=\pmb{\gamma}$ 
 and $P {L}={L} P$. 
\end{definition}
\begin{lemma}
The set of Eqs.\ \eqref{two} and the set of Eqs.\ \eqref{lse} have the same set of symmetries.
\end{lemma}
{\begin{proof}
We break the proof in two parts. We first show that (i) $P A = A P$ implies $P \hat{A}= \hat{A} P$, where the matrix  $\hat{A}=\{\hat{A}_{vw}\}$ has entries $\hat{A}_{vw}=A_{vw} \cos(\vartheta_v^*-\vartheta_w^*)$. Then show that (ii) $P \hat{A}= \hat{A} P$ implies  $P L=L P$.
Assume $P$ is a symmetry for the set of Eqs.\ (3). From Remark 1, it follows that for all $v,w \in \mathcal{V}$, $A_{vw}=A_{v'w'}$ and $A_{wv}=A_{w'v'}$ where $v'$ ($w'$) is the node $v$ ($w$) is mapped to by $P$. We are now going to show that also $\hat{A}_{vw}=\hat{A}_{v'w'}$ and $\hat{A}_{wv}=\hat{A}_{w'v'}$. If $v$ gets mapped into $v'$ they belong to the same cluster, and so also $w$ and $w'$, hence $\vartheta^*_v=\vartheta^*_{v'}$ and $\vartheta^*_w=\vartheta^*_{w'}$. Hence,  $\cos(\vartheta^*_v-\vartheta^*_w)=\cos(\vartheta^*_{v'}-\vartheta^*_{w'})$, which proves (i). To prove (ii) we just need to note that 
for two nodes $k$ and $l$ to be in the same cluster it must be necessarily verified that  $\sum_j A_{kj}=\sum_j A_{lj}$ and so also that $\sum_j \hat{A}_{kj}=\sum_j \hat{A}_{lj}$, see also \cite{NC,S_nchez_Garc_a_2020}. 
\end{proof}}


 The linear approximation yields this expression for the flow in a line $(i,j) \in \mathcal{E}$:
\begin{equation}\label{flow}
    F_{ij}=A_{ij}(\theta_j(t)-\theta_i(t)).
\end{equation}

In the rest of this paper, except for Sec.\ IV, we introduce the following assumption of homogeneous damping terms:
\begin{assumption} \label{a}
 All the damping terms are the same $\gamma_i=\gamma$, $i=1,...,n.$
\end{assumption}
 The regime in which all the $\gamma_i=\gamma$ is of interest because it has been shown to lead to the best properties in terms of stability of the synchronous solution \cite{motter2013spontaneous}. So, unless specified otherwise, in Secs.\ II and III,  we set $\gamma_i=\gamma=0.9$. In Sec.\ IV we will remove assumption \ref{a} and generalize our results to the case of arbitrary $\gamma$'s.


\begin{lemma}
 If the original matrix $\tilde A$ is symmetric, the spectrum of the matrix $L$ is real. 
  \end{lemma}

\begin{proof}
We note that $\hat{A}_{ij}=A_{ij} \cos(\vartheta^*_{i}-\vartheta^*_{j})= J_i \bar{A}_{ij}$, where the matrix $\bar{A}=\{\bar{A}_{ij}\}$ with entries $\bar{A}_{ij}=\tilde{A}_{ij} \cos(\vartheta^*_{i}-\vartheta^*_{j})$ is symmetric. By defining the symmetric matrix $\bar{L}=\{{\bar L}_{ij}\}$ with entries ${\bar L}_{ij}={\bar A}_{ij} -\delta_{ij} \sum_j {\bar A}_{ij}$, we can write $L=D \bar{L}$ and the matrix $D$ is a diagonal matrix with diagonal entries $D_{ii}=J_i$. By construction, both the matrix $L$ and the matrix $\bar{L}$ have sums of their rows equal to zero. However, the matrix $L$ is generically asymmetric, while the matrix $\bar{L}$ is symmetric.
We then write
the eigenvalue equation 
     $L\mathbf{v}= D \bar{L} \mathbf{v}= \lambda \mathbf{v}$,
      where $\mathbf{v}$ is an eigenvector and $\lambda$ is an eigenvalue for $L$.
 Then, we premultiply the above equation by $D^{-1/2}$ and obtain:
 \begin{equation}\nonumber
     D^{1/2}\bar{L}\mathbf{v}=\lambda D^{-1/2} \mathbf{v}.
      \end{equation}
Now, by setting $\mathbf{w}=D^{-1/2} \mathbf{v}$ we can write
\begin{equation}\nonumber
     D^{1/2}\bar{L} D^{1/2}\mathbf{w}=\lambda \mathbf{w},
     \end{equation}
 where the matrix $L'=D^{1/2}\bar{L}D^{1/2}$  is symmetric,  and thus the eigenvalues $\lambda$ are real. \\
 To conclude, $D\bar L$ is a special case of a generically asymmetric matrix with real spectrum. However, under generic conditions, the real eigenvalues of $L$ and $\bar L$ are not the same. 
 \end{proof}

Since the sums of the entries along all the rows of the matrix $L$ (of the matrix $\tilde L$) are zero, it follows that this matrix  has at least one eigenvalue equal zero. The multiplicity of the zero eigenvalue is equal to the number of connected components of the network. 
We now distinguish the two cases that the network is (i) connected and (ii) not connected. In case (i), the Laplacian matrix has a single zero eigenvalue $\lambda_1=0$ with associated eigenvector $[1,1,...,1]$ and the remaining eigenvalues $\lambda_2 \geq \lambda_3 \geq ....\lambda_n$ are all negative.  In case (ii), assume there are $c$ connected components. Then, after appropriate labeling of the nodes, the Laplacian matrix can be written as 
\begin{equation}
    L=\begin{bmatrix}L_{1} & & \\ & \ddots & \\ & & L_{c}\end{bmatrix},
\end{equation}
where $L_i$, $i=1,...,c$ is the Laplacian matrix associated to connected component $i$. It follows that the multiplicity of the zero eigenvalue of the Laplacian matrix $L$ is equal to $c$, $\lambda_1=...=\lambda_c=0$,  and the corresponding eigenvectors are each one associated with a connected component $i=1,...,c$; these $c$ eigenvectors are called the component vectors of the network connected components. The remaining eigenvalues $\lambda_{c+1} \geq \lambda_{c+2} \geq ....\lambda_n$ are all negative. \\

\begin{definition}
Component vectors are vectors that have entries $j$ equal to $1$ if node $j$ is in connected component $i$ and equal to $0$ otherwise.
\end{definition}

\begin{lemma}\label{diag} Assume the network is connected. The left eigenvector of the matrix $L=D \tilde L$ associated with the zero eigenvalue is equal to $[D_{11}^{-1},D_{22}^{-1},...,D_{nn}^{-1}]$. 
\end{lemma}
\begin{proof} By assumption, the matrix $L$ has only one zero eigenvalue. Consider the eigenvalue equation $\textbf{z}^T D \tilde{L}=\textbf{0}$, with left eigenvector $\textbf{z}$. The vector $\textbf{z}^T D$ is the left eigenvector of the matrix $\tilde L$ associated with the zero eigenvalue and therefore, it is equal to $[1,1,...,1].$ It follows that the entries of $\textbf{z}$ are the reciprocal of the entries on the main diagonal of the matrix $D$.
\end{proof}


We now consider the vector ${\pmb{\theta}}=[\theta_1,\theta_2,...\theta_n]$ and the vector ${\mathbf{p}}=[{p}_1,{p}_2,...,{p}_n]$, and rewrite Eq.\ \eqref{lse},
\begin{equation} \label{main}
    \ddot{\pmb{\theta}}(t) = -\gamma \dot{\pmb{\theta}}(t) + L \pmb{\theta}(t) + \mathbf{p}.
\end{equation}

We first assume the network is connected. We diagonalize $L$, $L=V \Lambda V^{-1}$, where $\Lambda=(\lambda_1,\lambda_2,...,\lambda_n)$ is the matrix of the eigenvalues of $L$ and $V=\{V_{ij}\}$ is the matrix of the eigenvectors of $L$.  
We multiply Eq. \eqref{two} on the left by $V^{-1}$ and by calling $\pmb{\eta}(t)=V^{-1} \pmb{\theta}(t)$ and $\mathbf{q}=V^{-1} \mathbf{p}$, we obtain,
\begin{equation} \label{diagmodes}
    \ddot{{\pmb{\eta}}}(t) = -\gamma \dot{{\pmb{\eta}}}(t) + \Lambda {\pmb{\eta}}(t) + {\mathbf{q}}, 
\end{equation}
which can be broken up into $n$ independent equations or `modes',
\begin{equation} \label{gen}
    \ddot{{{\eta}}}_i(t) = -\gamma \dot{{{\eta}}}_i(t) + \lambda_i {{\eta}_i}(t) + {{q}_i}, 
\end{equation}
$i=1,...,n$. {Modal decompositions similar to Eq. \eqref{gen} have been also obtained in \cite{bamieh2013price,guo2018graph,kettemann2016delocalization,tyloo2020primary,Powersys}.{Tyloo et al. \cite{tyloo2020primary} relax the constant inertia to damping ratio assumption and show that
their derivation is still valid with heterogeneous dynamical parameters.} In what follows we study the effects of network symmetries on the modal decomposition. We also show how the modes can be exploited in order to compute maximum flows over lines and in designing line capacities.} 

\begin{lemma}: Assume a balanced power grid, i.e., $\sum_i b_i=0$. Then, $q_1=0$. \\
\end{lemma}
\begin{proof}Consider $p_i=J_i b_i$ and $q_1=\mathbf{z^T} \mathbf{p}$, where $\mathbf{z}$ is the left eigenvector of $L$ associated with $\lambda_1$. \\
From Lemma \ref{diag}, we know that $\textbf{z}^T= [\frac{1}{J_1},\frac{1}{J_2},...\frac{1}{J_n}]$, so we get\\
\begin{equation}\nonumber
q_1=\sum_i\frac{1}{J_i}p_i\\
\end{equation}
and it follows that\\
\begin{equation}\nonumber
q_1=\sum_i \frac{1}{J_i}J_i b_i=\sum_i b_i=0.\\
\end{equation}
\end{proof}

Hence, for $i=1$, Eq.\ \eqref{gen} becomes
\begin{equation}\label{eq:single}
    \ddot{{{\eta}}}_1(t) + \gamma \dot{{{\eta}}}_1(t)  = 0,
\end{equation} 
from which we can see that for a large $t$, $\eta_1(t)$ approaches a constant, which depends on the initial conditions.

For $i>1$,
$\lambda_i<0$ ensures convergence of the corresponding $\eta_i(t)$ in Eq.\ \eqref{gen}  for large $t$.  Moreover, for $i=2,...,n$, Eq.\ \eqref{gen} can be written as,
\begin{equation}\label{UD}
    \ddot{{{\eta}_i}}(t) + 2 \zeta_i \omega_i \dot{{{\eta}_i}}(t) + {\omega_i}^2 {{\eta}_i}(t) = {{q}_i}, 
\end{equation}
where  $-\lambda_i=\omega_i^2$ and $\zeta_i={\gamma}/{(2 \omega_i)}$, $i=2,...,m$.
 Eq.\ \eqref{UD} is the equation of an underdamped ($0 < \zeta_i <1 $, $i=2,...,m$) or overdamped ($\zeta_i>1$, $i=m+1,...,n$) second order system forced by a step function of amplitude $q_i$.  Overdamped systems $i>m$ do not give rise to overshoots.  {The solution to \eqref{UD} can be written as $\eta_i(t)=\bar{\eta}_i(t)+\hat{\eta}_i(t)$, where $\bar{\eta}_i(t)$ is the free evolution and $\hat{\eta}_i(t)$ is the forced evolution. The free evolution decays exponentially in time. Thus, for large $t$, we can assume that $\eta_i(t) \simeq \hat{\eta}_i(t)$.\\
For overdamped systems, the forced solution is equal to,
\begin{equation}\label{OD}
    \eta_i(t)=\frac{q_i}{\omega_i^2}\Bigr[1+\frac{\chi_1e^{{-\omega_i\chi_2t}}-\chi_2e^{{-\omega_i\chi_1t}}}{2\sqrt{\zeta_i^2-1}}\Bigr]
\end{equation}\\
where $\quad \Bigr(\zeta_i-\sqrt{\zeta_i^2-1} \Bigr)=\chi_1, \quad \Bigr(\zeta_i+\sqrt{\zeta_i^2-1} \Bigr)=\chi_2$}.

For underdamped systems, the forced solution is equal to,
\begin{equation}\label{eta}
    {\eta}_i(t)=\frac{q_i}{\omega_i^2}\Bigl[1 - \frac{e^{-\zeta_i \omega_i t}}{\sqrt{1-\zeta_i^2}} \sin \Bigl(\omega_i \sqrt{1-\zeta_i^2 }t -\cos^{-1} \zeta_i \Bigr) \Bigr].
\end{equation}

{In realistic systems with constant damping and inertia, it has been found that all modes are underdamped and propagate through the whole system with $\frac{D}{\lambda_i}< \gamma^{-1}, \quad \forall i>1$ \cite{tyloo2020primary,Keyplayer_Tyloo}.  Thus, for the remainder of the paper, we will only focus on underdamped modes.}

The underdamped system in \eqref{eta} converges in steady state to
\begin{equation}
    {\eta}_i^{ss}=\frac{q_i}{\omega_i^2},
\end{equation}
{
and is upper and lower bounded,
\begin{equation}
   \eta_i^-(t) \leq \eta_i(t) \leq \eta_i^+(t),  
\end{equation}
where $\eta_i^{\pm}(t)=({\omega_i^2} \sqrt{1-\zeta_i^2})^{-1} ({q_i}\sqrt{1-\zeta_i^2} \pm {|q_i|} e^{- \zeta_i \omega_i t})$.
}

The peak time is equal to 
\begin{equation}
    t_i=\frac{\pi}{\omega_i \sqrt{1-\zeta_i^2}},
\end{equation}
and the peak value is given by
\begin{equation} \label{peak}
    {\eta}_i^{\mbox{peak}}=\frac{q_i}{\omega_i^2}\Bigl[1 + \exp({-\pi \zeta_i}/{\sqrt{1-\zeta_i^2}}) \Bigr].
\end{equation}
 The larger $q_i$ is, the higher is the peak. The smaller $\omega_i^2=-\lambda_i$ is, the higher is the peak. Finally, the smaller is $\zeta_i$, the higher is the peak. The eigenvalue $\lambda_2$ is associated with the smallest $\omega_2$ and also the smallest $\zeta_2$, and so is generically responsible for the largest peak.

As we know the $\eta_i(t)$'s, we can also compute the $\theta_i(t)$'s using the formula $\pmb{\theta}(t)=V \pmb{\eta}(t)$, or equivalently,
\begin{equation} \label{lincomb}
    \theta_i(t)= \sum_j V_{ij} \eta_j(t),
\end{equation}\\
It is important to note that the $V_{ij}$ can be either positive or negative, and so are the $q_i$.

In general, the peak of $\theta_i(t)$ will be more strongly affected by $\eta_2(t)$, followed by $\eta_3(t)$, $\eta_4(t)$, etc., but it will also depend on the terms $ V_{ij} q_j$. We can also compute the values of $\theta_i$ at steady state,
\begin{equation}\label{steadystate}
    \theta_i^{ss}= \sum_j V_{ij} \eta_j^{ss} = \sum_j V_{ij} \frac{q_j}{\omega_j^2}.
\end{equation}

From Eq. \eqref{flow}, the flows are equal to
\begin{equation}\label{floweta}
    F_{ij}(t)=\sum_{k=2}^n (V_{jk}-V_{ik}) \eta_k(t),
\end{equation}
where the summation starts from $k=2$ since for a connected network $V_{i1}=V_{j1}$ (the first column of the matrix $V$ is the eigenvector associated with the eigenvalue $0$.)
{Each of the terms in the summation on the right hand side of Eq.\ (21) can be either positive or negative. We  rewrite the right hand side of (21) as $\sum_{k+} (V_{jk}-V_{ik}) \eta_k(t)+\sum_{k-} (V_{jk}-V_{ik}) \eta_k(t)$, where the first summation is only over positive terms and the second summation is only over negative terms. The absolute flow in a line is then upper bounded by
\begin{equation} \label{ubounds}
|F_{ij}(t)| \leq \sum_{k+} (V_{jk}-V_{ik}) \eta_k^+(t)+\sum_{k-}|(V_{jk}-V_{ik}) \eta_k^-(t)|.
\end{equation}}

\subsection{Linear combinations of modes}

{
Equation \eqref{lincomb}  (equation \eqref{floweta}) shows that the individual nodal displacements $\theta_i(t)$ (the flows $F_{ij}(t)$) can be written  as linear combinations of the modes $\eta_i(t)$, $i=1,..,n$. This has immediate practical implications. For example, because each mode is both upper bounded and lower bounded, upper bounds on the absolute flows can be computed, see e.g., Eq.\ \eqref{ubounds}. However, an open question is how peaks and peak times can be computed for linear combinations of second order modes, such as those in Eqs.\ \eqref{lincomb} and \eqref{floweta}. The reason why this is challenging is that 
 peak times for different $\eta$s are different, thus a linear combination of the peak $\eta$ values will not return the peak flow. 
In this subsection we present a simple numerical technique to approximate the peak flows, based only on minimal knowledge about the individual modes.}

Our original assumption is that the flows $F_{ij}$ are small, which is the condition for approximating Eq. \eqref{two} with the linear swing equation \eqref{lse}. However, flows may increase as a result of a variety of perturbations, including internal power surges and outside attacks. It is also possible that the absolute flow $|F_{ij}|$ over a line may increase above the so-called line capacity, causing failure of that line and  possibly lead to other line failures in a cascade. A similar event would lead to violation of the small flow assumption.

Reference \cite{DIcascade} has pointed out the importance of the transient dynamic as it is possible that an absolute flow may exceed a line capacity transiently. Our analysis allows us to study the transient response of the swing equation dynamics under the assumption of small flows. Under this condition, we can compute the peak flows over each line, which can be used to design line capacities
 that are robust against the effects of transient perturbations of small entity.



Next we present a numerical technique to find the maximum absolute flow over a line. 
For large $\gamma$, there are two possible ways in which a flow in a line may evolve over time; \emph{\textbf{(1)} The absolute flow reaches its maximum value at the first peak time. \textbf{(2)} The absolute flow reaches its maximum value at steady state}. Thus, to find the maximum flow accurately, one needs to calculate both the first peak and the steady state value and pick the one which is largest in magnitude.  

In order to compute the peak times, we note that
from Eq. \eqref{floweta} we have that
\begin{equation}\label{flowtheta}
\dot{F}_{ij}=C_2\dot{\eta}_2+C_3\dot{\eta}_3+...C_n\dot{\eta}_n
\end{equation}
where for underdamped modes
\begin{equation}\label{etadot}
    \dot{\eta}_k=\frac{q_k}{\omega_k\sqrt{1-{\zeta_k}^2}}e^{-\zeta_k \omega_k t} \sin \omega_k \sqrt{1-{\zeta_k}^2}t,\hspace{3mm} k=m+1,...,n
\end{equation}
Since $\zeta_k \omega_k=\frac{\gamma}{2}$, the exponents in Eq. \eqref{etadot} are the same for $k=m+1,...,n$. Further, say $\varsigma_k=\omega_k \sqrt{1-{\zeta_k}^2}$. 
A necessary condition for the flow $F_{ij}$ to achieve  either a maximum or a minimum is 

\begin{equation}\label{sumsin}
{\sum_{k=m+1}^n} \frac{C_k q_k}{\varsigma_k} \sin{\varsigma_k t}=0.
\end{equation}

Since the values of $\varsigma_k$ are different for different $k$, we are unable to compute a closed form solution for \eqref{sumsin}. However,  a root-finding algorithm can be used to calculate the peak time and the maximum absolute flow through the line. For example, that can be done by computing an initial guess $\tau_{ij}$ and by iterating Newton's method towards convergence on the peak time.

For large enough $\gamma$ we expect the peak time to be the root of \eqref{sumsin} closest to the origin. 
{Then a good choice of the initial guess for the root-finding algorithm can be obtained by Taylor expanding \eqref{sumsin} up to third order about the origin and setting the third order expansion equal to zero. This leads to the following initial guess,}

\begin{equation}\label{rsa}
    \tilde{\tau}_{ij}=\sqrt{\frac{6 \sum_{k=m+1}^n C_k q_k }{\sum_{k=m+1}^n C_k q_k \varsigma_k^2}}.
\end{equation}
To demonstrate this, we consider the `bottleneck network' shown in Fig. \ref{fig1} with $\gamma=1.5$ {and $\textbf{p}=[0.3,-0.1,-0.1,0.1,-0.1,-0.2,0.1]^T$}. {We use
Eq.\ \eqref{rsa} to compute the initial guesses which are then used to approximate the first peak times and subsequently the first peak values that are shown in Table \ref{tab:SSFP}.} Table I also shows that for this example our approach based on the linearized swing equation (5) well approximates the maximum flow obtained by integration of the full nonlinear swing equation (3).

{For non large $\gamma$s however, a slightly more cumbersome approach has to be taken to find the initial guess. Peak $\eta$ values are first multiplied with their respective flow coefficients and then these values are cumulatively added in an ascending order of peak time. The peak time corresponding to the largest of these cumulatively added values can be used as the initial guess for the flow. It is important to note that the closer the peak times are, the more accurate this initial guess is and the farther apart they are, the more iterations will be required to converge to the actual solution.
}

{
We wish to emphasize that the approach described in this section can be applied to any linear combination of second order modes. For example,we can use this approach to approximate peak flows in quotient networks, Eq.\ \eqref{quot} and in error networks, Eq.\ \eqref{error}.}

\begin{table}
\caption{\label{tab:SSFP} {Steady state, first peak and maximum flow through lines for Fig. \ref{fig1}. The values for the second column are computed using Eq.\ \eqref{steadystate} and the values for the third column are computed using Eq.\ \eqref{rsa}, Eq. \eqref{eta} and Eq. \eqref{floweta}. The fourth and fifth columns are calculated by numerically solving the linear swing equation (5) and the non-linear swing equation (3), respectively.}}
\begin{center}
    \begin{tabular} {|l|l|l|l|l|l|l|}
    \hline
    &&& Linear & Non-Linear\\ \hline
          Line&	Steady&  First&	Max & Max\\ 
          &	    State&  Peak&  Flow &  Flow\\ \hline
          $\theta_2-\theta_1$& -0.1750&  -0.1536&  -0.1750& -0.1757 \\ \hline
          $\theta_3-\theta_1$& -0.1250&  -0.1466&    -0.1466& -0.1468\\ \hline
          $\theta_3-\theta_2$& 0.0500&  0.0000&    0.0500 & 0.0502\\ \hline 
          $\theta_4-\theta_2$& 0.0750&  0.0527&    0.0750  &0.0751\\ \hline          
          $\theta_5-\theta_2$& -0.2000&  0.0000&    -0.2000 & -0.2014\\ \hline  
          $\theta_4-\theta_3$& 0.0250& 0.0509 &  0.0509& 0.0509 \\\hline
          $\theta_6-\theta_5$& -0.1000&  0.0000&    -0.1000& -0.1002\\ \hline
          $\theta_7-\theta_5$& 0.0000&  0.0753&    0.0753& 0.0754 \\ \hline
          $\theta_7-\theta_6$& 0.1000&   0.1221&  0.1221& 0.1222 \\ \hline
           \end{tabular}
          
           \end{center}
           \end{table}

\begin{figure}
    \includegraphics[width=3.45in]{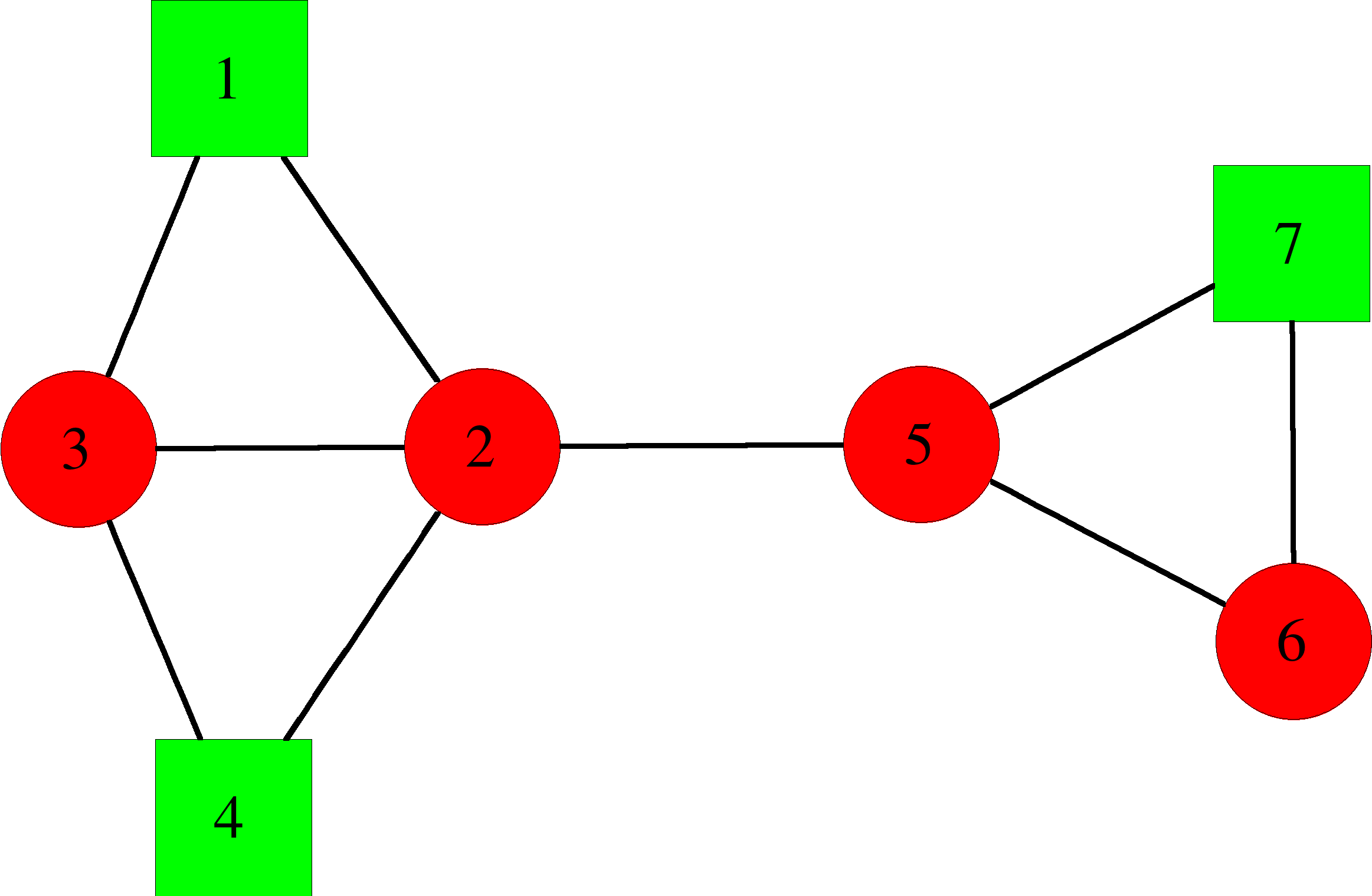}
    \caption{{Bottleneck network with 7 nodes. Red circles represent motor nodes while green squares represent generator nodes.}}
    \label{fig1}
\end{figure}


\section{Use of symmetry in the transient analysis of the linear swing equation}

Knowledge of the 
symmetries may be used to reduce the computational burden associated with the modal decomposition. In fact, we can apply the modal decomposition to a reduced `quotient network' where all the redundancies are eliminated, thereby making the analysis easier and faster.

 We  introduce the $n\times q$ indicator matrix $E$, where $n$ is the number of nodes and $q$ is the number of orbits or clusters. Each entry $E(i,j)=1$ if node $i$ belongs to the cluster $j$ and 0 otherwise. Another way to write the indicator matrix is the following, 
$E=[\mathbf{e}_1,\mathbf{e}_2,...,\mathbf{e}_q]$ where each vector $\mathbf{e_i}$ is an indicator vector, that is entry $j$ of vector $\mathbf{e}_i$ is 1 if node $j$ is in cluster $i$ and is 0 otherwise.

\begin{lemma}
 The set of states such that $\theta_l(t)=\theta_m(t)$ and $\dot{\theta}_l(t)=\dot{\theta}_m(t)$ for all $l,m \in S_i$, $i=1,...,q$, define
an invariant manifold \cite{NC}.
\end{lemma}

The full network with Laplacian matrix $L$ can  be transformed into its corresponding quotient network with Laplacian matrix $\tilde{L}$ by the transformation \cite{schaub2016graph},
\begin{equation}\nonumber
    \tilde{L}=((E^TE)^{-1}E^T)LE.\\
\end{equation} 

By pre-multiplying Eq. \eqref{main} by $((E^TE)^{-1}E^T)$, we obtain,

\begin{equation} \label{quot}
    \ddot{\tilde{\pmb{ \theta}}}(t) = -\gamma \dot{\tilde{\pmb{ \theta}}}(t) + \tilde{L} \tilde{\pmb{ \theta}}(t) + \tilde{\mathbf{ p}},
\end{equation}
where the $q$-dimensional vectors ${\tilde{\pmb{\theta}}}(t) = ((E^TE)^{-1}E^T) {\pmb{\theta}}(t)$ and $\tilde{\mathbf{p}}=[\tilde{p}_1,\tilde{p}_2,...,\tilde{p}_q]$ is equal to $\tilde{\mathbf{p}}= ((E^TE)^{-1}E^T) {\textbf{p}}$, .

\begin{definition} \label{rs}
We say that the power vector $\textbf{p}$ respects the symmetries if for all the nodes $i$ in the same cluster $S_j$, $p_i=\tilde{p}_j$, $j=1,...,q$. 
\end{definition}

In what follows we first consider the case that the power vector $\textbf{p}$ respects the symmetries and then consider the more general case that it does not.

\subsection{The {case that} the power vector $\mathbf{p}$ respects the symmetries}

When the power vector respects the symmetries, the forced evolution for the full network and the quotient network will be \emph{the same}, as we explain in what follows.


\begin{definition}\label{redundant_non}
A redundant eigenvector $\pmb{\rho}$ of the Laplacian matrix $L$ has sum of entries corresponding to each cluster equal to zero and a non redundant eigenvector $\pmb{\nu}$ of the Laplacian matrix have entries corresponding to each cluster that are all the same \cite{S_nchez_Garc_a_2020}. All the redundant  eigenvectors are orthogonal to $\mathbf{e}_i$, $i=1,..,q$. Contrarily, the non-redundant eigenvectors are not. That means $\mathbf{e}_i^T\pmb{\rho}=0$ and $\mathbf{e}_i^T\pmb{\nu} \neq 0$, $i=1,..,q$, where $\pmb{\rho}$ is any redundant vector, and $\pmb{\nu}$ any non redundant vector.
\end{definition}

From Definition \ref{redundant_non} we see that the matrix $L$ has $q$ non-redundant eigenvectors $\pmb{\nu}$'s and $(n-q) $ redundant eigenvectors $\pmb{\rho}$'s \cite{S_nchez_Garc_a_2020}.

\begin{remark}
{From Definition \ref{redundant_non} it follows that the modes \eqref{gen} can be  either redundant or nonredundant. The numer of redundant modes is equal to $(n-q)$ and the numer of nonredundant modes is equal to $q$.}
\end{remark}

All of the eigenvalues of the matrix $\tilde{L}$ are also eigenvalues of the matrix $L$ \cite{S_nchez_Garc_a_2020}. However, the matrix $L$ has additional eigenvalues that are not eigenvalues of $\tilde{L}$. When an eigenvalue of $L$ is also an eigenvalue of $\tilde{L}$, we say that it is `inherited' by the quotient network.
\begin{lemma}
Only the eigenvalues of $L$ associated with non-redundant eigenvectors are inherited by the quotient network. 
\end{lemma}
\begin{proof}

Consider the equation:\\
\begin{equation}
 \tilde{L}=(E^TE)^{-1}E^T  L  E=  (E^TE)^{-1}E^T  (VDV^T) E,   
\end{equation}
{where $V$ is the matrix of the eigenvectors of $L$ and $D$ is the matrix of eigenvalues of $L$.}

The matrix $V=[V_{\rho},V_{\nu}]$ where $V_\rho$ has all the redundant eigenvectors and $V_{\nu}$ has all the non-reduntant eigenvectors. Further, $D$ can be written as $D_\rho\oplus D_{\nu}$, where $D_\rho$ is a diagonal matrix with all the $(n-q)$ redundant eigenvalues and $D_{\nu}$ is a diagonal matrix with all the $q$ nonredundant eigenvalues. Now we see that $E^TV_\rho=0$ due to the orthogonality. Therefore,

\begin{equation}\label{Lq}
   \tilde{L}=(E^TE)^{-1}E^T  (V_{\nu}D_{\nu}V_{\nu}^T) E =(E^TE)^{-1}(HD_{\nu}H^T), 
\end{equation}
where $H=E^TV_{\nu}$ is a square $q$-dimensional matrix.\\

\end{proof}

From Eq. \eqref{Lq}, we can see that only non-redundant eigenvalues and eigenvectors are inherited by the quotient network. Moreover since the power vector $\mathbf{p}$ respects the symmetries, the  entries of the vector $\tilde{{p}}_j=p_i$ for any node $i$ in cluster $S_j$ (see Definition \ref{rs}).


From $\sum_j{p}_j=0$ it follows that $\sum_{j=1}^m= n_j{\tilde p}_{j}=0$, where $n_j$ is the number of nodes in each cluster.

\begin{remark}
In the case that the network is connected, 
the eigenvector of the Laplacian matrix $[1 1 ... 1]$, associated with the zero eigenvalue, is non-redundant. So, the zero eigenvalue is inherited by the quotient network. Further, it can also be proved under generic conditions that in the case of a disconnected network, the quotient network inherits $c$ zero eigenvalues from the full network, where $c$ is the number of components. 
\end{remark}

\begin{theorem}
Assume the network is connected.
The entry of the vector $\tilde{\mathbf{q}}$ corresponding to the zero eigenvalue  is 0. 
\end{theorem}
\begin{proof}
Consider the equation:
\begin{equation}
    \mathbf{w^T}\tilde{\mathbf{p}}=0 
\end{equation}

Let $\mathbf{v}$ be a right eigenvector and $\mathbf{w^T}$ be a left eigenvector of $\tilde{L}$. We know the right eigenvector associated with the zero eigenvalue is $\mathbf{v}=[1 1...1]$. Which implies, \\
\begin{equation*}
    \tilde{L}\mathbf{v}=0
    \end{equation*}
    
    \begin{equation*}
(E^TE)^{-1}E^TLE\mathbf{v}=0
\end{equation*}

By left multiplying both sides of the above equation by $(E^TE)$, we obtain,\\
\begin{equation*}
E^T L E\mathbf{v}=0    
\end{equation*}
Indicating that the matrix $E^TLE$ must have $\mathbf{v}=[1 1...1]$ as its right eigenvector associated with the zero eigenvalue. A left eigenvector $\mathbf{w}^T$ associated with zero eigenvalue has to satisfy the equation:
\begin{equation}
 \mathbf{w^T}(E^TE)^{-1}E^TLE=0
\end{equation}
This also means $\mathbf{w^T}(E^TE)^{-1}$ is a left eigenvector with associated eigenvalue 0 for $(E^TLE)$ which is a symmetric matrix. Because $E^TLE$ is symmetric, its left eigenvectors and right eigenvectors, corresponding to the same eigenvalues, are the same.\\
Ergo,\\
\begin{equation}
  \mathbf{w^T}(E^TE)^{-1}=\mathbf{v^T}    
\end{equation}
Recall that
$\mathbf{v^T}=[1 1....1]$. Moreover, the matrix
$(E^TE)$ is a diagonal matrix that has the number of nodes in each cluster $n_1, n_2,$ etc on the main diagonal. It follows that 
$\mathbf{w^T}=[n_1 n_2 ...]$.
Since the entries of the vector $\mathbf{w^T}$ are the number of nodes in each cluster, the product $\mathbf{w^T} \tilde{\mathbf{p}}$ is equal to the sum of all the powers in the full network which is zero by the assumption that the full network is balanced.
\end{proof}


In order to demonstrate this quotient transformation, we will use the $n=7$ node graph with $6$ edges shown in Fig. \ref{fig:cubic}, for which
\renewcommand{\path}{AllSetup/yNetwork}

\begin{figure}
   \includegraphics[width=2.6in,center]{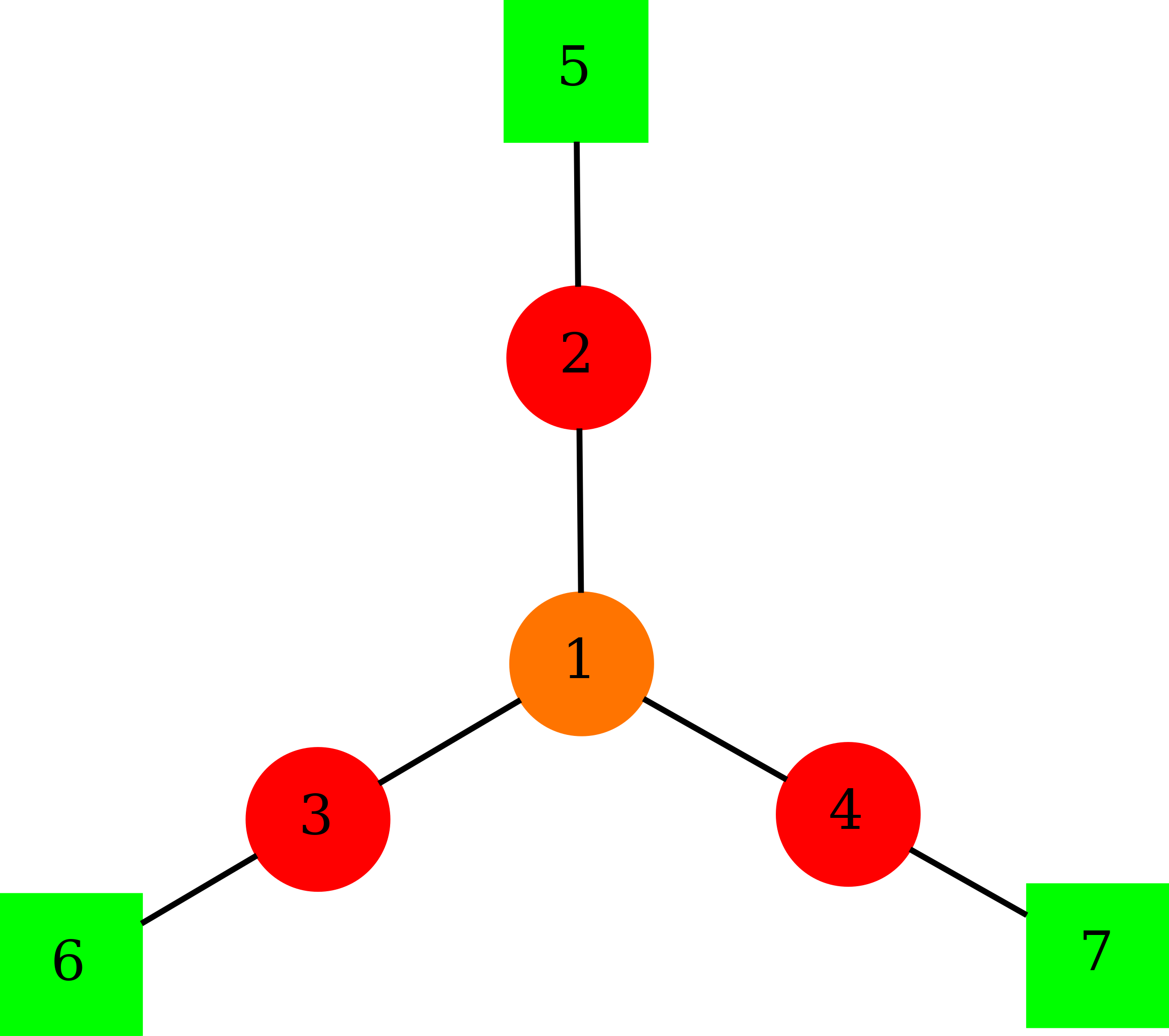}
    \caption{7-node network. Generator nodes are indicated with squares and motor nodes with circles. All nodes corresponding to the same cluster or orbit of the symmetry group are colored the same.}
    \label{fig:cubic}
\end{figure}

\begin{center}
$L$=$\begin{bmatrix} 
-3&	1&	1&	1&	0&	0&	0\\
1&	-2&	0&	0&	1&	0&	0\\
1&	0&	-2&	0&	0&	1&	0\\
1&	0&	0&	-2&	0&	0&	1\\
0&	1&	0&	0&	-1&	0&	0\\
0&	0&	1&	0&	0&	-1&	0\\
0&	0&	0&	1&	0&	0&	-1 
\end{bmatrix}$ 
\end{center}
     with eigenvalues $0,
    -0.38,
    -0.38,
    -1.58,
    -2.61,
    -2.61,
    -4.41$. The power vector $\mathbf{p}$ and the transformed power vector $\mathbf{q}$ are also shown, 
\begin{equation*}
    \begin{tabular}{c c}
     


        $\mathbf{p}$= $\begin{bmatrix}
       -0.3\\
-0.3\\
-0.3\\
-0.3\\
0.4\\
0.4\\
0.4
        \end{bmatrix}$&

        $\mathbf{q}$= $\begin{bmatrix}
   0.00\\
    0.00\\
   0.00\\
   -0.8895\\
    0.00\\
    0.00\\
    0.2208
        \end{bmatrix}$

        \end{tabular}
\end{equation*}
 As can be seen the power vector $\mathbf{p}$ respects the symmetries. 
\newline

Figure \ref{flotime} shows the flows vs. time over all the network lines (forced evolution). Due to the network symmetries, several flows are superimposed on top of each other (and thus indistinguishable). 
\begin{figure}
    \includegraphics[width=3.45in]{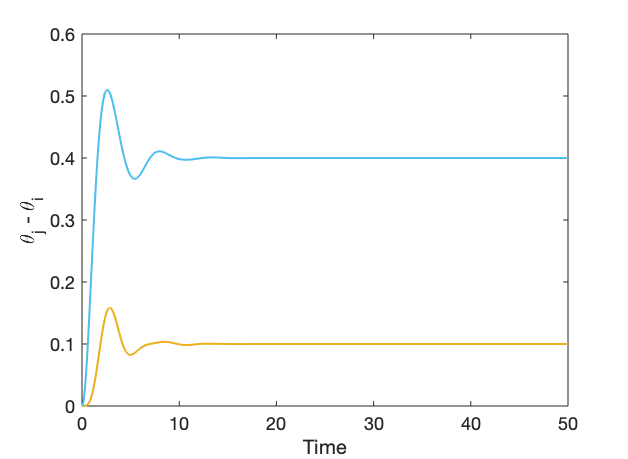}
    \caption{{Flows vs. time for network in Fig.\ \ref{fig:cubic}. The network has a total of 6 lines. 
    The forced evolution of the flows over each one of these lines is plotted. However, due to the symmetries, triplets of flows are superimposed on top of each other.}}
    \label{flotime}
\end{figure}

 For the network in Fig.\ 2, $S_1=\{1\}$, $S_2=\{2,3,4\}$, $S_3=\{5,6,7\}$. This corresponds to the following indicator matrix:
\begin{center}
    \begin{tabular}{c c c}
        E= $\begin{bmatrix}
      1&	0&	0\\
0&	1&	0\\
0&	1&	0\\
0&	1&	0\\
0&	0&	1\\
0&	0&	1\\
0&	0&	1
    \end{bmatrix}$&
\end{tabular}
\end{center}
The  Laplacian matrix $\tilde{L}$ for the quotient network is,

\begin{figure}
    \includegraphics[width=3.5in,center]{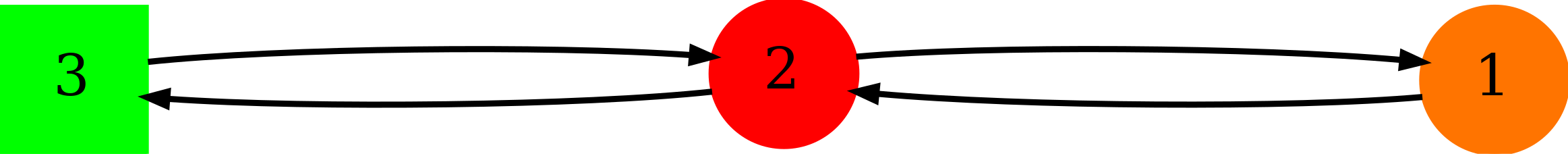}
    \caption{Quotient Network corresponding to the network in Fig. \ref{fig:cubic}.}
    \label{fig:my_label}
\end{figure}

\begin{center}
    \begin{tabular}{c}
        
         $\tilde{L}$= $\begin{bmatrix}
           -3&    3&     0\\
    1&     -2&    1\\
     0&    1&     -1
        \end{bmatrix}$. \\
        
        \end{tabular}
        \end{center}
     with eigenvalues $0, -1.58, -4.41$. Moreover,
       
      \begin{center}
    \begin{tabular}{c}
        $\tilde{V}_{r}$= $\begin{bmatrix}
       0.898&	0.731&	0.577\\
    -0.423&	0.345&	0.577\\
    0.124&	-0.589&	0.577
        \end{bmatrix}$
        \end{tabular}\vspace{0.1in}
       
       \begin{tabular}{c}
          $\tilde{V}_{l}$= $\begin{bmatrix}
     0.562&   0.337&    0.229\\
    -0.794&   0.476&    0.688\\
   0.233&    -0.813&    0.688
        \end{bmatrix}$\\
        \end{tabular}
        \end{center}

\begin{center}
    \begin{tabular}{c c}

        $\tilde{\mathbf{p}}$= $\begin{bmatrix}
        -0.3\\
-0.3\\
0.4
 \end{bmatrix}$&

        $\tilde{\mathbf{q}}$= $\begin{bmatrix}
       0.1873\\
-0.6402\\
0
        \end{bmatrix}$,
        \end{tabular}
\end{center}
where
 $\tilde{V}_{r}$ and $\tilde{V}_{l}$ represent the matrix of right and the left eigenvectors of $\tilde{L}$ respectively. The vectors $\tilde{\mathbf{p}}$ and $\tilde{\mathbf{q}}$ are the power vector and the transformed power vector of the quotient network respectively.\\
As seen in the Fig. \ref{fig:my_label}, the quotient network has three nodes (each representing a cluster) and two lines. Unlike the full network, this network is directed and weighted. {Under the assumption that the vector $\textbf{p}$ respects the symmetries,} a study of the forced response of the quotient network provides complete information about the forced response of the full network.
If we plot the flows of the quotient network,  we still get Fig. \ref{flotime} but with only one curve representing all the nodes in a cluster.

\subsection{The {case that} the power vector $\mathbf{p}$ does not respect the symmetries}

When the power vector does not respect the symmetries, we can still reduce the dynamics in the quotient network form \eqref{quot}. In this case
the power for each node of the quotient network will be the average of the powers of the nodes in each cluster and the forced evolution of the quotient network is the average of the forced evolutions of the full network within each cluster. 

As an example, consider the network in Fig.\ 2 studied previously but with power vector,
 
\begin{center}

        $\mathbf{p}$= $\begin{bmatrix}
        -0.3\\
-0.3\\
-0.4\\
-0.3\\
0.5\\
0.4\\
0.4
\end{bmatrix}$
\end{center}
which does not respect the symmetries. The forced 
evolutions for the full network and for the quotient network are given in Figs. \ref{pnores}(a) and \ref{pnores}(b). The color code consistent with color of nodes in Figs.\ \ref{fig:cubic} and \ref{fig:my_label}. All the green (red) curves in Fig.\ \ref{pnores}(a) average to the one green (red) curve in Fig.\ \ref{pnores}(b).


\begin{figure}
    \includegraphics[width=3.8in]{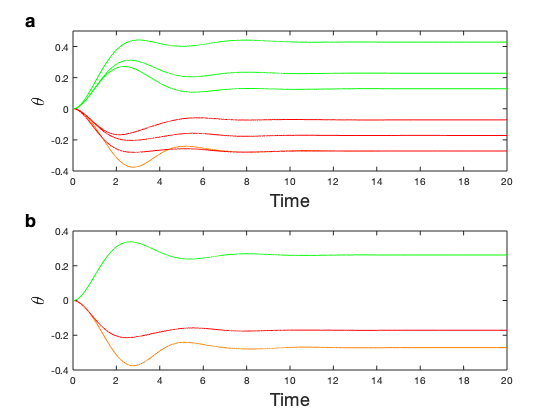}
    \caption{{\textbf{a}:$\theta$ vs. time for the network in Fig.\ \ref{fig:cubic}, \textbf{b}: $\theta$ vs. time for quotient network in Fig.\ \ref{fig:my_label}. Color code consistent with color of nodes in Figs.\ \ref{fig:cubic} and \ref{fig:my_label}.}}
    \label{pnores}
\end{figure}

An important observation is that the quotient network removes redundant modes but these modes might still produce overshoots in the full network which will be ignored by the quotient network. 

{Since the forced evolution obtained from the quotient network is not the same as that of the full network, a relevant problem is to characterize deviations between the time evolutions of the nodes in the same cluster and the time evolution of the associated quotient node.  
\begin{definition}
Define the deviations as $\epsilon_{i}(t)=(\theta_i(t)-\tilde{\theta}_{{i*}}(t))$, $i=1,2,...,n$, where ${i*}$ is the cluster node $i$ belongs to, $\sum_{i \in S_j} \epsilon_{i}(t)=0$.
\end{definition}
The dynamics of the deviations is provided by the redundant modes only. Without loss of generality assume the redundant modes \eqref{gen} are labeled $i=1,...,(n-q).$ Then
\begin{equation} \label{error}
    \epsilon_i(t)=\sum_{j=1}^{n-q} V_{ij} \eta_j(t).
\end{equation}
The modal decomposition of the forced response Eq.\ \eqref{error} can be used to compute the transient deviation overshoot dynamics, peak times, upper bounds, etc. We can also approximate the peaks using the numerical approach of Sec.\ IIA.
\begin{lemma}
In the case the power vector $\textbf p$ respects the symmetries, the redundant modes \eqref{gen} have $q_i=0$, hence their forced response is zero.
\end{lemma}
\begin{proof}
This follows trivially from Definitions \ref{rs} and \ref{redundant_non}.
\end{proof}}

{\subsection{Large Network Example}
We now consider a large network, shown in Figure \ref{Chilefull}, based off the Chilean power-grid topology with $n=218$ nodes, out of which $94$ are motors and $124$ are generators. This is the reduced version achieved by the Star-Mesh transformation \cite{Star_mesh_Bedrosian} of the original Chilean power-grid  \cite{Chile_Kim}.
\subsubsection{Powers {that} respect the symmetries}
We first choose the power vector such  that it respects the symmetries;  for simplicity we set all generator nodes to have a power of {0.0081} and all motor nodes to have a power of {-0.0106}. For this network, we have 29 non-trivial clusters; clusters with more than one node, and 96 trivial clusters; clusters with only one node.
Using this information, we can transform the full network into the quotient network shown in Figure \ref{Chilequot}.  After the transformation, each cluster becomes a node in the quotient network with 157 nodes, out of which 74 are generator nodes and 83 are motor nodes. As already discussed previously, if the vector $\textbf{p}$ respects the symmetries,  the forced response of the quotient network coincides the forced response of the full network {(not shown here for the sake of brevity.)}

\subsubsection{Powers {that} do not respect the symmetries} If the power vector does not respect the symmetries, the quotient network time evolution provides the average forced evolution of the nodes in each cluster of the full network. The deviation between the dynamics of the full network and the quotient network can be characterized using Eq.\ \eqref{error}. To demonstrate this we first choose a cluster where the nodes do not respect the symmetries. The cluster is chosen containing 7 nodes: 178, 181, 182, 182, 183, 184, 185, 207,  with associated powers {[0.0181, 0.0081, 0.0081, 0.0081, 0.0081, 0.0081, 0.0081].}} 
Now, using analogous equations to Eq. \eqref{sumsin} and \eqref{rsa}, we can obtain the maximum displacement deviation of each node $\epsilon_i$. Since we are finding the maximum of displacements rather than flows, we replace the coefficients $C_i$ with $V_i$ and get,
\begin{equation}\label{errorsum}
  {\sum_{k=1}^{n-q}} \frac{V_k q_k}{\varsigma_k} \sin{\varsigma_k t}=0.
\end{equation}
\begin{equation}\label{rsaerror}
    \tilde{\tau}_{ij}=\sqrt{\frac{6 \sum_{k=1}^{n-q} V_k q_k }{\sum_{k=1}^{n-q} V_k q_k \varsigma_k^2}}.
\end{equation}
Using the technique mentioned in Section IIA and Eq.\ \eqref{errorsum} and \eqref{rsaerror}, we can easily calculate maximum displacement deviation. The  time evolution of this deviation for our selected cluster is shown in Fig.\ \ref{E}.  

\begin{figure}
    \includegraphics[width=3.8in]{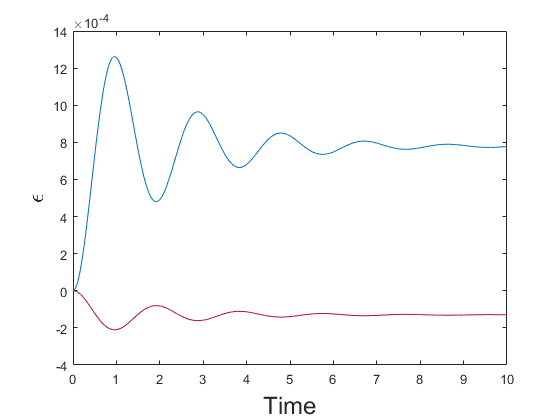}
    \caption{{Evolution of the displacement deviation of the quotient network node with respect to the original nodal displacements. The blue curve is for node 178, the red curve is for nodes 181, 182, 183, 184, 185, 207.}}
    \label{E}
\end{figure}

\begin{figure}
    \includegraphics[width=3.45in]{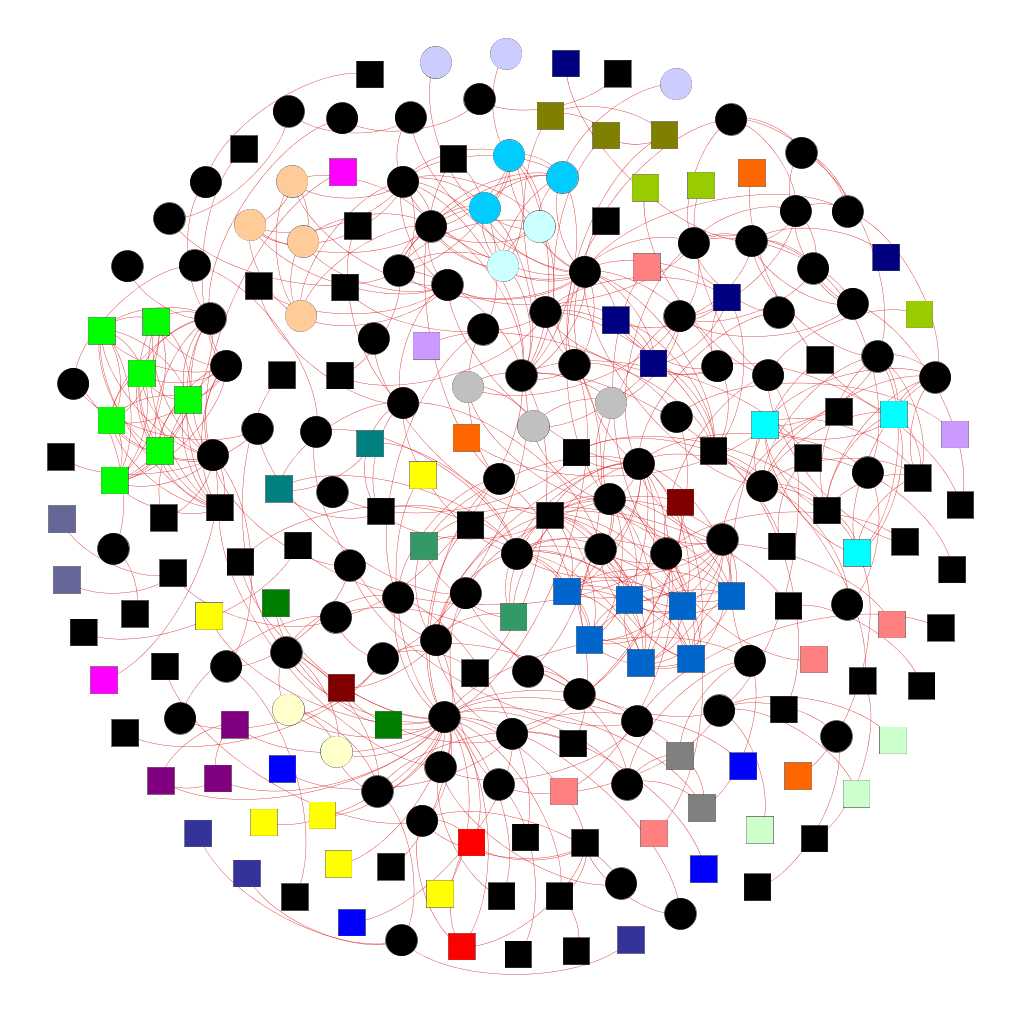}
    \caption{Chilean power-grid Network. Squares indicate generator nodes and circles indicate motor nodes. Black nodes are in trivial clusters whereas nodes with the same color (different from black) are in non-trivial clusters.}
    \label{Chilefull}
\end{figure}

\begin{figure}
    \includegraphics[width=3.45in]{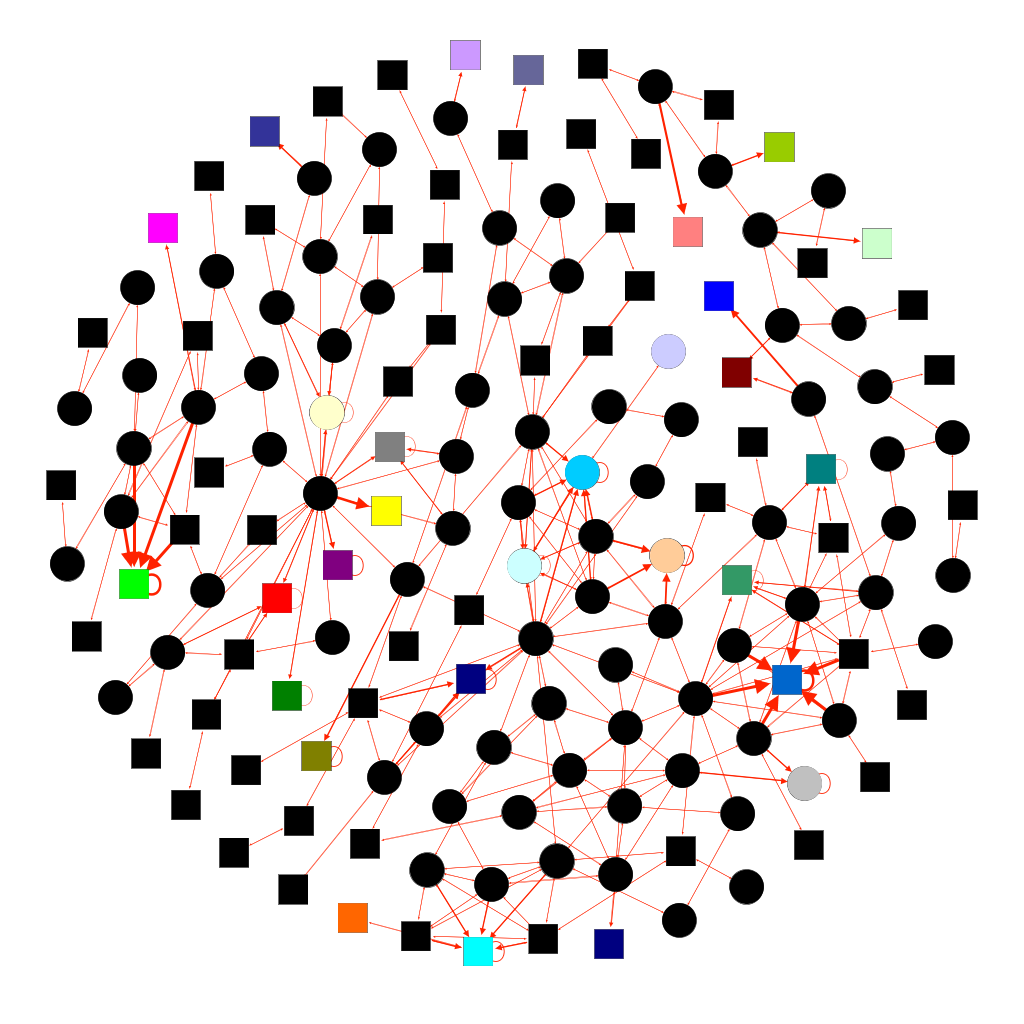}
    \caption{Quotient network for Chilean Power-grid topology. Squares indicate generator nodes and circles indicate motor nodes. Black nodes are in trivial clusters whereas nodes with the same color (different from black) are in non-trivial clusters. Thickness of the arrows are indicative of the weight of the edges.}
    \label{Chilequot}
\end{figure}

\section{Effects of heterogeneity in the $\gamma_i$ terms}

In this section we remove the assumption of homogeneous damping terms $\gamma_i=\gamma$. Note that the symmetries of the swing equation are defined in the generic case of heterogeneous (arbitrary) $\gamma$'s, see Definition \ref{defsym}. Here we explain how our previous derivations need to be modified in the case of heterogeneous $\gamma$'s.

Equation \eqref{main} is rewritten,
\begin{equation} \label{mainH}
    \ddot{\pmb{\theta}}(t) = -\Gamma \dot{\pmb{\theta}}(t) + L \pmb{\theta}(t) + \mathbf{p},
\end{equation}
where the $n$-dimensional diagonal matrix $\Gamma=diag({\gamma}_1,{\gamma}_2,...,{\gamma}_n)$.
Analogously, the quotient network Equation \eqref{quot} becomes,
\begin{equation} \label{quotH}
    \ddot{\tilde{\pmb{ \theta}}}(t) = -\tilde{\Gamma} \dot{\tilde{\pmb{ \theta}}}(t) + \tilde{L} \tilde{\pmb{ \theta}}(t) + \tilde{\mathbf{ p}}, 
\end{equation}
where the $q$-dimensional diagonal matrix $\tilde{\Gamma}=diag(\tilde{\gamma}_1,\tilde{\gamma}_2,...,\tilde{\gamma}_q)$ and $\tilde{\gamma}_i$ is the damping ratio of all the nodes in cluster $S_i$.



A separate consideration is needed for the study of the deviations in the case that (i) the $\gamma$'s are heterogeneous and (ii) the power vector $\textbf p$ does not respect the symmetries. To deal with such a case, we introduce the irreducible representations of the symmetry group $\mathcal{G}$ \cite{NC}.

From knowledge of the group of symmetries $\mathcal{G}$, we can compute the irreducible representations (IRRs) of the symmetry group of the network. This defines a transformation $T$ into the so-called IRR coordinate system (see Ref. \cite{NC}). The transformation matrix $T$ is orthogonal.  Each one of the rows of the matrix $T$ is associated with a specific cluster. If a row of the matrix $T$ is associated with cluster $k$, it means all the $i$ entries of that row are zero for $i$ not in cluster $S_k$. The first $q$ rows of the matrix $T$ are parallel to the $q$ nonredundant eigenvectors $\pmb \nu$.

Therefore, the matrix $T$ can be cast in the following block diagonal form,
\begin{equation}
T=\bigoplus_{k=1,..,q} T_k,
\end{equation}
where each block $T_k$ is an $n_k$-dimensional square matrix associated with cluster $k$.
Premultiplying Eq.\ \eqref{mainH} by $T$, we obtain,
\begin{equation} \label{mainB}
    \ddot{\pmb{\xi}}(t) = -\Gamma \dot{\pmb{\xi}}(t) + B \pmb{\xi}(t) + \mathbf{r},
\end{equation}
where $\pmb{\xi}= T \pmb{\theta}$, $\mathbf{r}= T \mathbf{p}$, the block -diagonal matrix $B=T L T^T$ and the matrix $\Gamma$ is unaltered by application of the matrix $T$. The latter property follows from the block-diagonal structure of the matrix $T$ and the observation that $\Gamma$ can be recast in a similar form 
\begin{equation}
    \Gamma=\bigoplus_{k=1,...,q} I_{n_k} \tilde{\gamma}_k,
\end{equation}
where $I_{n_k}$ is the $n_k$-dimensional identity matrix.

The transformed $n \times n$ block diagonal matrix $B= T A T^T$ is a direct sum $\oplus_{u=1}^U \hat{B}_u$, where $\hat{B}_u$ is a (generally complex) $p_u \times p_u$ matrix
with $p_u$ the multiplicity of the $u$th IRR in the permutation
representation, $U$ the number of IRRs present and $d_u$
the dimension of the $u$th IRR, so that
$\sum_u d_u p_u= n$.
The matrix $T$ contains information on which perturbations affecting different clusters get mapped  to different IRRs \cite{siddique2018symmetry}.

Due to the block-diagonal structure of the matrix $B$, Eq.\ \eqref{mainB} can be decoupled into a number of lower  dimensional equations, where each equation corresponds to a block of the matrix $B$.
 
 There is one representation (labeled $u=1$) which we call {trivial} and has dimension $d_1=q$. The associated block of the matrix $B$ corresponds to the dynamics of the quotient network. Hence, the trivial representation is associated with all the clusters. However, it is possible that other IRR representations are only associated with some of the clusters (not all of them.) Each one of these other representations $u>1$ describes the deviation dynamics of either an isolated cluster or a group of intertwined clusters \cite{NC}. A simple interpretation of isolated vs. intertwined clusters is the following. If a cluster is isolated a perturbation affecting the power of any one of its nodes will not affect the deviation dynamics of other clusters. On the contrary, when a set of two or more clusters are intertwined,  a perturbation affecting the power of any of the nodes in a cluster will affect the deviation dynamics of the remaining clusters in the set.

\subsection*{Example A} For this example, we consider the dynamics of  Eq.\ \eqref{lse}, with adjacency matrix corresponding to the network in Fig.\ \ref{fig:cubic}, the power vector {$\mathbf{p}=[-0.3,-0.3,-0.3,-0.4,0.4,0.4,0.5]^T$} and the matrix $\Gamma=diag(3,2,1,1,2,1,1)$.
{By the definition of symmetries (Definition 3), nodes belonging to the same cluster have the same $\gamma$ value.} In this case, due to the presence of heterogeneity in $\gamma$, there are two non-trivial clusters, each containing 2 nodes: $\{3,4\}$ and $\{6,7\}$. All other nodes are in trivial clusters. Moreover, the power vector $\mathbf{p}$ does not respect the symmetries.

We compute the $IRRs$ of the symmetry group. The transformation matrix $T$ is equal to, 
\small
\begin{center}
    \begin{tabular}{c}
        $T$= $\begin{bmatrix}
 1&	0&	0&	0&	0&	0&	0\\
0&	1&	0&	0&	0&	0&	0\\
0&	0&	\frac{1}{\sqrt{2}}&	\frac{1}{\sqrt{2}}&	0 &	0 &	0\\
0&	0&	0&	0&	1&	0&	0\\
0&	0&	0&	0&	0&	\frac{1}{\sqrt{2}}&	\frac{1}{\sqrt{2}}\\
0&	0&	\frac{1}{\sqrt{2}}&	-\frac{1}{\sqrt{2}}&	0&	0&	0\\
0&	0&	0&	0&	0&	\frac{1}{\sqrt{2}}&	-\frac{1}{\sqrt{2}}&
       \end{bmatrix}$
       \end{tabular}
       \end{center}\normalsize
       With this information, we can calculate the block-diagonal matrix $B=TLT^T$, ${\textbf r}=T{\textbf p}$, and $\Gamma_{orth}=T \Gamma T^T$,
       \small
       \begin{center}
       \[ 
         B = \left[\begin{array}{@{}ccccc|cc@{}}
         
         -3&   1&   \sqrt{2}&         0&         0&     0&         0\\
   1&    -2&         0&   1&         0&         0&         0\\
   \sqrt{2}&         0&   -2&         0&   1&    0&         0\\
         0&   1&         0&    -1&         0&         0&         0\\
         0&         0&   1&         0&    -1&   0&         0\\ \hline
         0&         0&    0&         0&   0&    -{2}&    {1}\\
         0&         0&   0&         0&    0&    {1}&    -{1} 
        \end{array}.\right]
        \]

          $\mathbf{r}$= $\begin{bmatrix}
-0.3&
   -0.3&
   -0.495&
    0.4&
    0.636& \vline
     {1/10\sqrt{2}}&
    {-1/10\sqrt{2}}
        \end{bmatrix}^T$. \\

          \[
             \Gamma_{orth}= \left[\begin{array}{@{}ccccc|cc@{}}
       3&        0&         0&         0&         0&         0&         0\\
         0&    2&        0&         0&         0&         0&         0\\
         0&         0&    1&      0&         0&    0&       0\\
         0&         0&         0&    2&     0&         0&         0\\
         0&         0&         0&         0&    1&       0&         0\\ \hline
         0&         0&         0&         0&         0&          {1}&       0\\
         0&         0&         0&         0&    0&     0&     {1}\\
        \end{array}. \right]
        \]
        \end{center}
        \normalsize

The transverse  (bottom right) block of the matrix $B$ represents the deviation dynamics for the two non-trivial clusters, which are intertwined. Note that the elements of $\mathbf r$ relative to this block are given by $r_{6}=\frac{(p_3 - r_4)}{\sqrt{2}}$ and $r_{7}=\frac{(p_6 - r_7)}{\sqrt{2}}$. We also note that $\Gamma_{orth}$ for the transverse block is the same as $\Gamma$ as explained in the theory. 
After diagonalizing the transverse block, its dynamics decouples into two independent modes,
{
\begin{equation} \label{kappa}    \ddot{\pmb{\kappa}}=-I\dot{\pmb{\kappa}}+\begin{bmatrix}
   -0.3820&   0\\
         0&    -2.6180
       \end{bmatrix}{\pmb{\kappa}}+\begin{bmatrix}
      0.02298\\
     -0.09732
       \end{bmatrix}
\end{equation}}


       


Now, using the technique derived in Section IIA, we can use Eq.\ \eqref{kappa} to calculate the maximum deviations from the quotient network. We can find the initial guess analogously  to Eq.\ \eqref{rsaerror}.
 $\tilde{\tau}$ can then be used to find the time of maximum error by solving the following equation for $t$ until convergence: \\
\begin{equation}\label{time_err}
{\sum_{k=m+1}^n} \frac{V_k r_k}{\varsigma_k} \sin{\varsigma_k t}=0.
\end{equation}
The obtained time $t$ can then be used to solve \eqref{lincomb} which gives us the deviations in each non-trivial cluster as seen in Table \ref{tab:error}. {{Table \ref{tab:error} also shows that for this example our approach based on the linearized swing equation (5) well approximates the maximum error obtained by integration of the full nonlinear swing equation (3).}}

\begin{table}
\small
\caption{\label{tab:error} {Steady state, first peak and maximum errors for the two non-trivial clusters. The values for the second column are obtained using Eq. \eqref{steadystate} and the values for the third column are obtained using Eq. \eqref{rsaerror}, Eq. \eqref{time_err} and Eq. \eqref{lincomb}. {The fourth and fifth columns are calculated by numerically solving the linear swing equation (5) and the non-linear swing equation (3), respectively.}}}
\begin{center}
    \begin{tabular} {|l|l|l|l|l|}
    \hline
           &  &  & Linear & Non-Linear\\\hline
          Deviation&	                    Steady&  First&	Max&Max \\ 
          &	    State&  Peak&  Error& Error\\ \hline
          $\frac{\theta_4-\theta_3}{\sqrt{2}}$&  0&  -0.0314&  -0.0314& -0.0323 \\ \hline
          $\frac{\theta_7-\theta_6}{\sqrt{2}}$&  0.0707& 0.0707&   0.0713 & 0.0792\\ \hline
           \end{tabular}
          
           \end{center}
           \end{table}
\normalsize
\subsection*{Example B}  For this example, we consider the Chilean power-grid described in Section IIIC and pictured in Fig.\ \ref{Chilefull}. The power vector is chosen so that it does not respect the symmetries. 
The coefficients $\gamma_i$ are chosen such that all nodes have $\gamma_i=1$, except for $\gamma_{116}=2$. For this network we computed the $IRR$ representations and the block-diagonal matrix $B$. This allowed us to study the deviation from the quotient network dynamics for each one of the $29$ nontrivial clusters. We did not find any intertwined clusters. Next, we focus on a cluster comprised of the three nodes $\{77,79,81\}$, represented as three cyan squares on the right side of Fig.\ 7. The transverse block corresponding to that cluster is reduced in a similar form to Eq.\ (41),
 {
 \begin{equation} \label{kappa1}    \ddot{\pmb{\kappa}}=-I\dot{\pmb{\kappa}}+\begin{bmatrix}
   -9.123&         0\\
    0.0000&   -0.876
       \end{bmatrix}{\pmb{\kappa}}+\begin{bmatrix}
     -0.00502\\
0.00643
       \end{bmatrix}
\end{equation}}
 
Similarly to Example A, Eq.\ (43) can be mapped to the deviation of the individual nodal displacement from the quotient displacement.  This can be done by computing a linear combination of the uncoupled modes (43) and applying the technique developed in Sec.\ IIA. Results of this computation are presented in Table \ref{tab:error2}, which shows our ability to accurately predict the maximum deviations for all three nodes $\{77,79,81\}$ in the cluster.

\begin{table}
 \caption{\label{tab:error2} {Steady state, first peak and maximum errors for the cluster containing nodes $\{77,79,81\}$. The values for the second column are obtained using Eq. \eqref{steadystate} and the values for the third column are obtained using Eq. \eqref{rsaerror}, Eq. \eqref{time_err} and Eq. \eqref{lincomb}. The fourth and fifth columns are calculated by numerically integrating the linear and non linear ODE respectively. $\theta_q$ refers to the displacement of the quotient node which is given by $\theta_q=\frac{\theta_{77}+\theta_{79}+\theta_{81}}{3}$} } 
 
 \begin{center}
 \small
 \begin{tabular} {|l|l|l|l|l|}
    \hline
    & & & Linear & Non-Linear \\ \hline
          Deviation&	                    Steady&  First&	Max & Max \\ 
          &	    State&  Peak&  Error & Error\\ \hline
          $\theta_{77}-\theta_q$& 0.00083&  0.001181& 0.01175 &0.001182 \\ \hline
         $\theta_{79}-\theta_q$&  -0.000416& -0.000591&   -0.000588 & -0.000577\\ \hline 
           $\theta_{81}-\theta_q$&  -0.000416& -0.000591& -0.000588 & -0.000577  \\ \hline 
           \end{tabular}
          
           \end{center}
           \end{table}
       
\normalsize

\section{Conclusions}

{In this paper we have studied how the presence of network symmetries affects the swing equation dynamics. We have first introduced the nonlinear swing equation and then modeled the propagation of small perturbations via the linearized swing equation. We have then shown that the nonlinear and the linear swing equation have the same symmetries. These symmetries can be reduced to provide an essential description of the dynamics in terms of a 'quotient network'.\\
We have  introduced a decomposition of the linear swing equation dynamics into independent modes, each of which corresponds to a forced second order system. This allowed us to characterize several transient effects such as overshoots, peaks, peak times, upper bounds etc. We have classified the symmetries into two cases based on whether the power vector respects the symmetries or not and presented techniques to obtain a full characterization of the transient dynamics and/or error quantification.\\  
In the case in which the power demanded and supplied by different nodes respects the symmetries, the quotient network completely describes the forced evolution of the full network, otherwise it provides the forced evolution {averaged over the nodes in each cluster}. We have also introduced the error dynamics which describes how the quotient network time evolutions deviate from the time evolutions of the individual nodes inside the clusters. This error dynamics can be written as a linear combination of the `redundant modes'.\\
Finally, we have presented how the symmetry analysis can be applied to networks with heterogeneous $\gamma$. In order to study the \textit{deviation dynamics} in the case that the $\gamma$'s are arbitrary and the power vector does not respect the symmetries, we have introduced the irreducible representation of the symmetry group. Furthermore, though not discussed in the paper, it is possible to apply our techniques to the case of different types of forcing terms, such as sinusoidal forcing; they can also be generalized to other network models, such as the effective network model (EN) or the structure preserving model (SP) \cite{powersync_latora}.}

\section*{Acknowledgement}
This work was supported by the National Science Foundation through grant (Grant No.
1727948.) The authors thank Cesar Ornelas for collaborating on an earlier version of this paper, Aranya Chakraborrti  for insightful discussions and Matteo Lodi for helping with the calculations of the irreducible representations.

\begin{IEEEbiography}[{\includegraphics[width=1in,height=1.25in,clip,keepaspectratio]{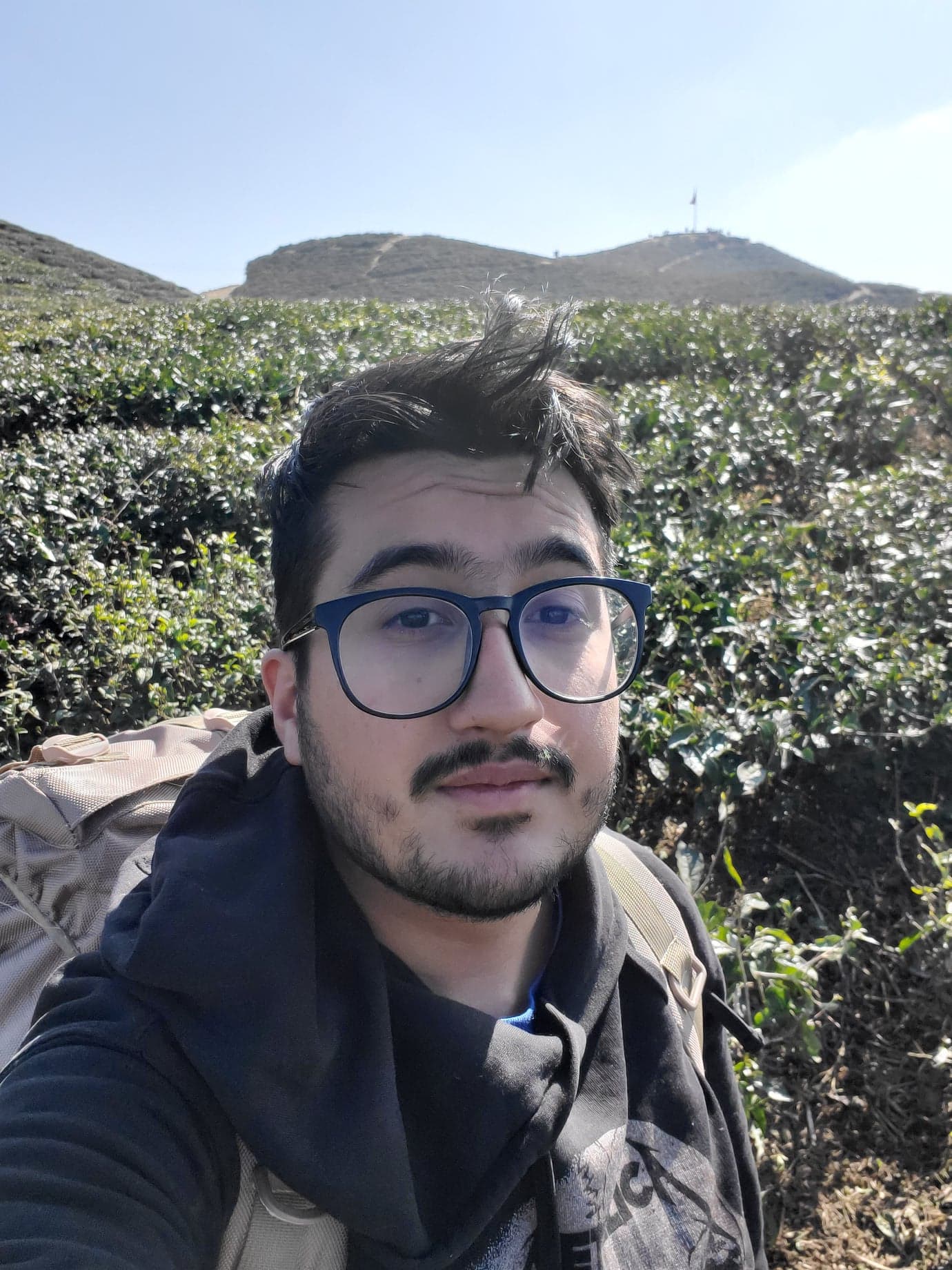}}]{Kshitij Bhatta}
 received a bachelor’s degree
in Mechanical Engineering in 2020 from the University
of New Mexico (US). He is currently a Master's student in Mechanical Engineering at the University of New Mexico. His research interests include non-linear dynamic systems, analysis of complex networks and control systems with feedback in robotic and automotive applications.     
\end{IEEEbiography}


\begin{IEEEbiography}[{\includegraphics[width=1in,height=1.25in,clip,keepaspectratio]{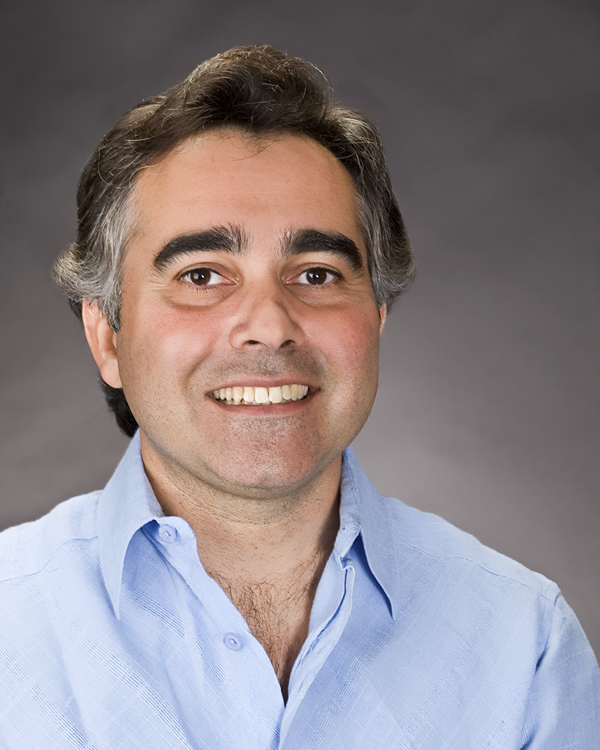}}]
{Majeed M. Hayat} received the M.S. and the Ph.D. degrees in Electrical and Computer Engineering from the University of Wisconsin-Madison in 1988 and 1992, respectively. He is currently a Professor and Department Chair of Electrical and Computer Engineering at Marquette University, Milwaukee, WI, USA. He was Associate Editor of Optics Express (Photodetectors and Image Processing) from 2004 to 2010 and Associate Editor and member of the Conference Editorial Board for the IEEE Control Systems Society. From 2010 to 2013 he was the Chair of the topical committee on Photodetectors, Sensors, Systems, and Imaging of the IEEE Photonics Society. From 2014 to 2018 he served as Associate Editor for the IEEE Transactions on Parallel and Distributed Systems. His research activities cover a broad range of topics including resilience and reliability of interdependent cyber-physical systems, dynamical modeling of cascading phenomena with applications to power systems, avalanche photodiodes, statistical communication theory, signal and image processing, algorithms for spectral and radar sensing and imaging, optical communication, and networked computing. He is a recipient of the National Science Foundation Early Faculty Career Award (1998). Dr. Hayat has authored or co-authored over 108 peer-reviewed journal articles and 134 conference papers (over 5,500 citations, H-Index: 37), and has fourteen issued patents, six of which have been licensed.
\end{IEEEbiography}

\begin{IEEEbiography}[{\includegraphics[width=1in,height=1.25in,clip,keepaspectratio]{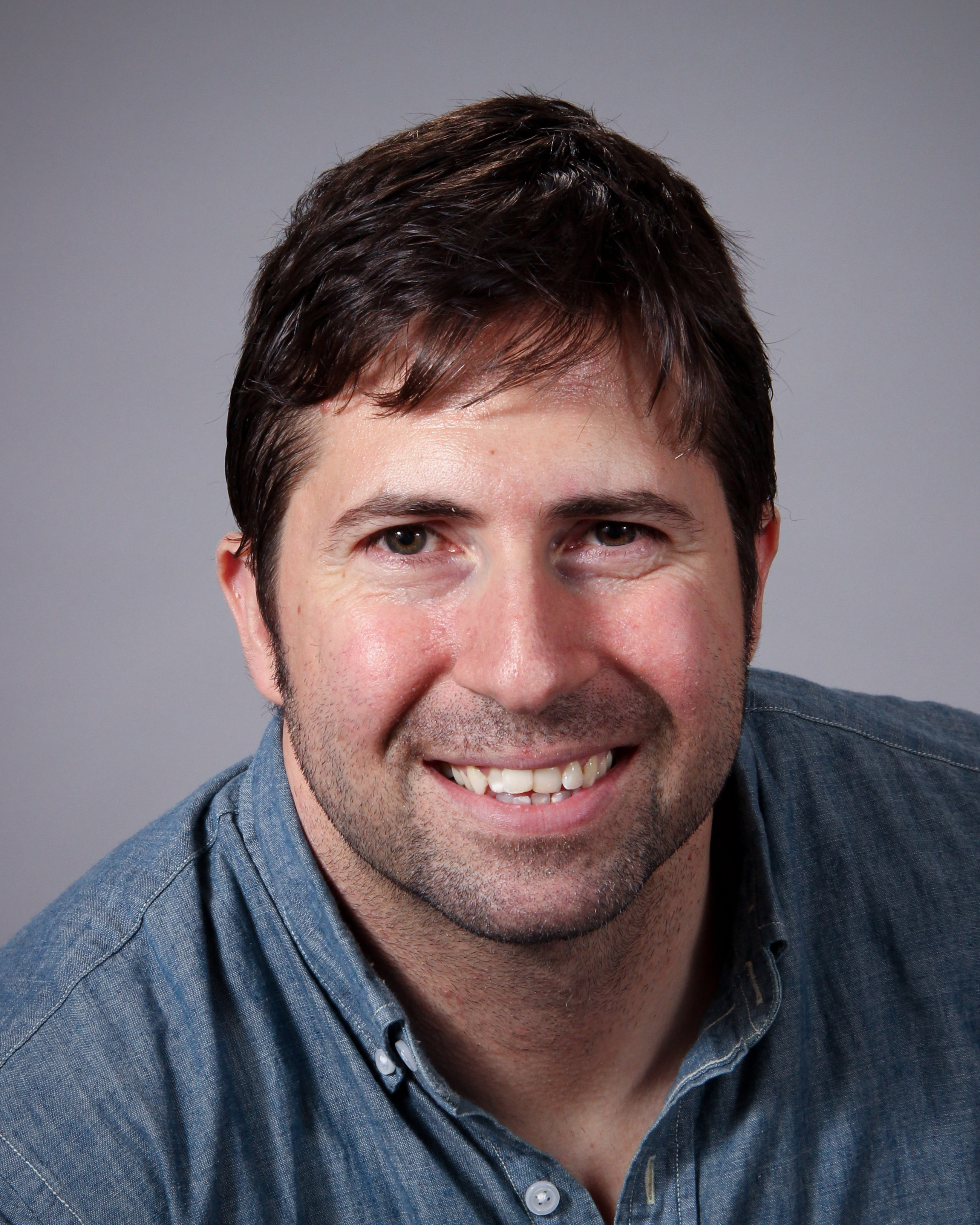}}]{Francesco Sorrentino}
 received a master’s degree
in Industrial Engineering in 2003 and a
Ph.D. in Control Engineering in 2007 both from the University
of Naples Federico II (Italy). His
expertise is in dynamical systems and controls,
with particular emphasis on nonlinear dynamics
and adaptive decentralized control. His work includes
studies on dynamics and control of complex
dynamical networks and hypernetworks,
adaptation in complex systems, sensor adaptive
networks, coordinated autonomous vehicles operating in a dynamically
changing environment, and identification of nonlinear systems. He is interested
in applying the theory of dynamical systems to model, analyze,
and control the dynamics of complex distributed energy systems such
as power networks and smart grids. Subjects of current investigation
are evolutionary game theory on networks (evolutionary graph theory),
the dynamics of large networks of coupled neurons, and the use of
adaptive techniques for dynamical identification of communication delays
between coupled mobile platforms. He has published more than 50
papers in international scientific peer reviewed journals.
\end{IEEEbiography}

\bibliographystyle{plain}
\newcommand{\noop}[1]{}

\end{document}